\documentclass[12pt]{article}

\usepackage{amsmath, amsthm, amssymb}
\usepackage{enumerate}
\newtheorem{theorem}{}[section]
\newtheorem{lemma}[theorem]{}

\newtheoremstyle{styleclaim}{}{}{\itshape}{}{}{}{ }{(\thmnumber{#2})}

\theoremstyle{styleclaim}

\newcommand{\cref}[1]{(\ref{#1})}

\begin{document}

\title{Three-coloring graphs with no induced seven-vertex path II : using a triangle}
\date{\today}
\author{Maria Chudnovsky\thanks{Princeton University, Princeton, NJ 08544, USA. E-mail: mchudnov@math.princeton.edu. Partially supported by NSF grants IIS-1117631, DMS-1001091 and DMS-1265803.}, Peter Maceli\thanks{Columbia University, New York, NY 10027, USA. E-mail: plm2109@columbia.edu.}, and Mingxian Zhong\thanks{Columbia University, New York, NY 10027, USA. E-mail: mz2325@columbia.edu.}}

\maketitle

\begin{abstract}

In this paper we give a polynomial time algorithm which determines if a given graph containing a triangle and no induced seven-vertex path is $3$-colorable, and gives an explicit coloring if one exists. This is the second paper in a series
of two. The first one, \cite{PART1} is also submitted to this journal.
In \cite{Tfree,PART1}, a polynomial time algorithm is given for three-coloring triangle-free graphs with no induced seven-vertex path. Combined, this shows that three-coloring a graph with no induced seven-vertex path can be done in polynomial time, thus answering a question of \cite{problem}.

\end{abstract}

\section{Introduction}
	
We start with some definitions. All graphs in this paper are finite and simple. Let $G$ be a graph and $X$ be a subset of $V(G)$. We denote by $G[X]$ \textit{the subgraph of $G$ induced by $X$}, that is, the subgraph of $G$ with vertex set $X$ such that two vertices are adjacent in $G[X]$ if and only if they are adjacent in $G$. We denote by 
$G\setminus X$ the graph $G[V(G) \setminus X]$. If $X=\{v\}$ for some $v \in V(G)$, we write $G \setminus v$ instead of $G \setminus \{v\}$.
Let $H$ be a graph. If $G$ has no induced subgraph isomorphic to $H$, then we say that $G$ is \textit{$H$-free}. For a family $\mathcal{F}$ of graphs, we say that $G$ is \textit{$\mathcal{F}$-free} if $G$ is $F$-free for every $F\in \mathcal{F}$. If $G$ is not $H$-free, then \textit{$G$ contains $H$}. If $G[X]$ is isomorphic to $H$, then we say that \textit{$X$ is an $H$ in $G$}. 

For $n\geq 0$, we denote by $P_{n+1}$ the path with $n+1$ vertices, that is, the graph with distinct vertices $\{p_0,p_1,...,p_n\}$ such that $p_i$ is adjacent to $p_j$ if and only if $|i-j|=1$. We call the set
$\{p_1, \ldots, p_{n-1}\}$ the {\em interior} of $P$. For $n\geq 3$, we denote by $C_n$ the cycle of length $n$, that is, the graph with distinct vertices $\{c_1,...,c_n\}$ such that $c_i$ is adjacent to $c_j$ if and only if $|i-j|=1$ or $n-1$. When explicitly describing a path or a cycle, we always list the vertices in order. Let $G$ be a graph. When $G[\{p_0,p_1,...,p_n\}]$ is the path $P_{n+1}$, we say that \textit{$p_0-p_1-...-p_n$ is a $P_{n+1}$ in $G$} or just a \textit{ path}, when there is no danger of confusion. Similarly, when $G[\{c_1,c_2,...,c_n\}]$ is the cycle $C_n$, we say that \textit{$c_1-c_2-...-c_n-c_1$ is a $C_n$ in $G$}. We also refer to a cycle of length three as a \textit{triangle}. A {\em clique} in a graph is a set of pairwise
adjacent vertices. A {\em stable set} is a set of vertices no two of which are
adjacent.

A \textit{$k$-coloring} of a graph $G$ is a mapping $c:V(G)\rightarrow \{1,...,k\}$ such that if $x,y\in V(G)$ are adjacent, then $c(x)\ne c(y)$. For
$X \subseteq V(G)$, we define by $c(X)=\bigcup_{x \in X}\{c(x)\}$.
If a $k$-coloring exists for a graph $G$, we say that  $G$ is \textit{$k$-colorable}. The COLORING problem is determining the smallest integer $k$ such that a given graph is $k$-colorable, and it was one of the initial problems R.M.Karp \cite{KARP} showed to be NP-complete. For fixed $k\geq 1$, the $k$-COLORING problem is deciding whether a given graph is $k$-colorable. Since Stockmeyer \cite{STOCK} showed that for any $k\geq 3$ the $k$-COLORING problem is NP-complete, there has been much interest in deciding for which classes of graphs coloring problems can be solved in polynomial time. In this paper, the general approach that we consider is to fix a graph $H$ and consider the $k$-COLORING problem restricted to the class of $H$-free graphs. 

We call a graph \textit{acyclic} if it is $C_n$-free for all $n\geq 3$. The \textit{girth} of a graph is the length of its shortest cycle, or infinity if the graph is acyclic. Kami\'nski and Lozin \cite{CYCLES} proved:

\begin{lemma}\label{HFREE}

For any fixed $k,g\geq 3$, the $k$-COLORING problem is NP-complete for the class of graphs with girth at least $g$.

\end{lemma} 

\noindent As a consequence of \ref{HFREE}, it follows that if the graph $H$ contains a cycle, then for any fixed $k\geq 3$, the $k$-COLORING problem is NP-complete for the class of $H$-free graphs. The \textit{claw} is the graph with vertex set $\{a_0, a_1, a_2, a_3\}$ and edge set $\{a_0a_1, a_0a_2, a_0a_3\}$. A theorem of Holyer \cite{LINE} together with an extension due to Leven and Galil \cite{LINE1} imply the following:

\begin{lemma}\label{CLAW}

If a graph $H$ contains the claw, then for every fixed $k\geq 3$, the $k$-COLORING problem is NP-complete for the class of $H$-free graphs. 

\end{lemma}

\noindent Hence, the remaining problem of interest is deciding the $k$-COLORING problem for the class of $H$-free graphs where $H$ is a fixed acyclic claw-free graph. It is easily observed that every connected component of an acyclic claw-free graph is a path. And so, we focus on the $k$-COLORING problem for the class of $H$-free graphs where $H$ is a connected  acyclic claw-free graph, that is, simply a path. Ho\`{a}ng, Kami\'nski, Lozin, Sawada, and Shu \cite{P5} proved the following:

\begin{lemma}

For every $k$, the $k$-COLORING problem can be solved in polynomial time for the class of $P_5$-free graphs. 

\end{lemma}

\noindent Additionally, Randerath and Schiermeyer \cite{P6} showed that:

\begin{lemma}

The $3$-COLORING problem can be solved in polynomial time for the class of 
$P_6$-free graphs. 

\end{lemma}

\noindent In \cite{P6} and \cite{problem} the question of the complexity of 
$3$-coloring  $P_7$-free graphs was posed.

On the other hand, Huang \cite{4P7} recently showed that:

\begin{lemma}

The following problems are NP-complete:

\begin{enumerate}

\item The $5$-COLORING problem is NP-complete for the class of $P_6$-free graphs. 

\item The $4$-COLORING problem is NP-complete for the class of $P_7$-free graphs. 

\end{enumerate}

\end{lemma}

For our purposes, it is convenient to consider the following more general coloring problem. A \textit{palette $L$} of a graph $G$ is a mapping which assigns each vertex $v\in V(G)$ a finite subset of $\mathbb{N}$, denoted by $L(v)$. A \textit{subpalette} of a palette $L$ of $G$ is a palette $L'$ of $G$ such that $L'(v)\subseteq L(v)$ for all $v\in V(G)$. We say a palette $L$ of the graph $G$ has \textit{order} $k$ if $L(v)\subseteq \{1,...,k\}$ for all $v\in V(G)$. Notationally, we write $(G,L)$ to represent a graph $G$ and a palette $L$ of $G$. We say that a $k$-coloring $c$ of $G$ is a \textit{coloring of $(G,L)$} provided $c(v)\in L(v)$ for all $v\in V(G)$. We say $(G,L)$ is \textit{colorable}, if there exists a coloring of $(G,L)$. We denote by $(G,\mathcal{L})$ a graph $G$ and a collection $\mathcal{L}$ of palettes of $G$. We say $(G,\mathcal{L})$ is \textit{colorable} if $(G,L)$ is colorable for some $L\in\mathcal{L}$, and $c$ is a {\em coloring} of $(G,\mathcal{L})$ if $c$ is a coloring of $(G,L)$ for some $L\in\mathcal{L}$. 

Let $G$ be a graph. We denote by $N_G(v)$ (or by $N(v)$ when there is no danger of confusion) the set of neighbors of $v$ in $G$.
Given $(G,L)$, consider a subset $X,Y\subseteq V(G)$ . 
We say that we \textit{update the palettes of the vertices in $Y$ with respect to $X$}
(or simply \textit {update $Y$ with respect to $X$}), if for all $y\in Y$ we set $$L(y)=L(y)\setminus (\bigcup_{u\in N(y)\cap X \text{ with } |L(u)|=1}L(u)).$$ 
When $Y=V(G)$ and $X$ is the set of all vertices $x$ of $G$ with $|L(x)|=1$,
we simply say that we \textit{update $L$}.
Note that updating can be carried out in time $O(|V(G)|^2)$. By reducing to an instance of $2$-SAT, which Aspvall, Plass and Tarjan \cite{PSAT} showed can be solved in linear time, Edwards \cite{2SAT} proved the following:

\begin{lemma}\label{check} There is an algorithm with the following specifications:

\bigskip

{\bf Input:} A palette $L$ of a graph $G$ such that $|L(v)|\leq 2$ for all $v\in V(G)$.

\bigskip

{\bf Output:} A coloring of $(G,L)$, or a determination that none exists.

\bigskip

{\bf Running time:} $O(|V(G)|^2)$.

\end{lemma}

\noindent Let $G$ be a graph. A subset $S$ of $V(G)$ is called \textit{monochromatic} with respect to a given coloring $c$ of $G$ if $c(u)=c(v)$ for all $u,v\in S$. Let $L$ be palette of $G$, and $X$ a set of subsets of $V(G)$.
We say that $(G,L,X)$  is {\em colorable} if there is a coloring $c$ of $(G,L)$ such that $S$ is monochromatic with respect to $c$ for all $S \in X$. 
A triple $(G',L',X')$ is a \textit{restriction} of $(G,L,X)$ if $G'$ is an 
induced subgraph of $G$, $L'$ is a subpalette of  $L|_{V(G')}$, and $X'$ is a 
set of subsets of $V(G')$ such that if $S \in X$ then $S \cap V(G') \in X'$.
Let $\mathcal{P}$ be a set of restrictions of $(G,L,X)$.
We say that $\mathcal{P}$ is {\em colorable} if at least one element of $\mathcal{P}$ is colorable.  If $\mathcal{L}$ is a set of palettes of $G$, we
write $(G,\mathcal{L},X)$ to mean the set of restrictions $(G,L',X)$ where
$L' \in \mathcal{L}$. The proof of \ref{check} is easily modified to obtain the 
following generalization \cite{2SAT+}:

\begin{lemma}\label{checkSubsets} There is an algorithm with the following specifications:

\bigskip

{\bf Input:} A palette $L$ of a graph $G$ such that $|L(v)|\leq 2$ for all $v\in V(G)$, together with a set $X$ of subsets of $V(G)$.

\bigskip

{\bf Output:} A coloring of $(G,L,X)$, or a determination that none exists.

\bigskip

{\bf Running time:} $O(|X||V(G)|^2)$.

\end{lemma}

\bigskip

A subset $D$ of $V(G)$ is called a \textit{dominating set} if every vertex 
in $V(G)\setminus D$ is adjacent to at least one vertex in $D$.
Applying \ref{check} yields the following general approach for $3$-coloring a graph. Let $G$ be a graph, and suppose $D\subseteq V(G)$ is a dominating set. Initialize the order 3 palette $L$ of $G$ by setting $L(v)=\{1,2,3\}$ for all $v\in V(G)$. Consider a fixed $3$-coloring $c$ of $G[D]$, and let $L_c$ be the subpalette of $L$ obtained by updating the palettes of the vertices in $V(G)\setminus D$ with respect to $D$. By construction, $(G,L_c)$ is colorable if and only if the coloring $c$ of $G[D]$ can be extended to a $3$-coloring of $G$.  Since $|L_c(v)|\leq 2$ for all $v\in V(G)$,  \ref{check} allows us to  efficiently test
if $(G,L_c)$ is colorable. Let $\mathcal{L}$ to be the set of all such palettes $L_c$ where $c$ is a $3$-coloring of $G[D]$. 
It follows that $G$ is 3-colorable if and only if $(G,\mathcal{L})$ is colorable. Assuming we can efficiently produce a dominating set $D$ of bounded size, since there are at most $3^{|D|}$ ways to 3-color $G[D]$, it follows that we can efficiently construct $\mathcal{L}$ and test if  $(G,\mathcal{L})$ is colorable, and so we can decide if $G$ is 3-colorable in polynomial time. This method figures prominently in the polynomial time algorithms for the 3-COLORING problem for the class of $P_\ell$-free graphs where $\ell\leq 5$. However, this approach needs to be modified when considering the class of $P_\ell$-free graphs when $\ell\geq 6$, since a dominating set of bounded size may not exist. Very roughly, the techniques used in this paper may be described as such a modification.

In \cite{Tfree,PART1}, the following was shown:

\begin{lemma}\label{half1}There is an algorithm with the following specifications:

\bigskip

{\bf Input:} A $\{P_7,C_3\}$-free graph $G$.

\bigskip

{\bf Output:} A $3$-coloring of $G$, or a determination that none exists.

\bigskip

{\bf Running time:} $O(|V(G)|^7)$.

\end{lemma}

In this paper, we consider the case when the input graph contains a triangle and prove the following:

\begin{lemma}\label{half2}There is an algorithm with the following specifications:

\bigskip

{\bf Input:} A $P_7$-free graph $G$ which contains a triangle.

\bigskip

{\bf Output:} A $3$-coloring of $G$, or a determination that none exists.

\bigskip

{\bf Running time:} $O(|V(G)|^{24})$.

\end{lemma}

Together, \ref{half1} and \ref{half2} give:

\begin{lemma}\label{4C}There is an algorithm with the following specifications:

\bigskip

{\bf Input:} A $P_7$-free graph $G$.

\bigskip

{\bf Output:} A $3$-coloring of $G$, or a determination that none exists.

\bigskip

{\bf Running time:} $O(|V(G)|^{24})$.

\end{lemma}

Given a graph $G$ and disjoint subsets $A$ and $B$ of $V(G)$, we say that $A$ is
{\em complete} to $B$ if every vertex of $A$ is adjacent to every vertex of
$B$, and that $A$ is {\em anticomplete} to $B$ if every vertex of $A$ is 
non-adjacent to every vertex of $B$. If $|A|=1$, say $A=\{a\}$, we write
``$a$ is complete (or anticomplete) to $B$'' instead of ``$\{a\}$ is complete
(or anticomplete) to $B$''.

Here is a brief outline of our algorithm \ref{half2}. We take advantage of the simple fact that all three-colorings of a triangle are the same (up to permuting colors), and, moreover, starting with the coloring of a triangle, the colors of certain other vertices are forced. In this spirit, we define a \textit{tripod} in a graph $G$ as a triple $(A_1,A_2,A_3)$ of disjoint subsets of $V(G)$ such that

\begin{itemize}

\item $A_1\cup A_2\cup A_3=\{a_1,...,a_m\}$,

\item $a_i\in A_i$ for $i=1,2,3$,

\item $a_1-a_2-a_3-a_1$ is a triangle in $G$, and 

\item letting $\{i,j,k\}=\{1,2,3\}$, for every $s\in\{1,...,m\}$, if $a_s\in A_i$, then $a_s$ has a neighbor in $A_j \cap \{a_1,..,a_{s-1}\}$ and a neighbor in $A_k \cap \{a_1,..,a_{s-1}\}$.

\end{itemize}

\noindent Let $G$ be a $P_7$-free graph which contains a triangle. The first step of the algorithm is to choose a maximal tripod $(A_1, A_2, A_3)$ in $G$. It is easy to see that in every $3$-coloring of $G$, each of the sets $A_1,A_2,A_3$ is monochromatic, thus if one of $A_1,A_2 ,A_3$ is not a stable set, the algorithm stops and outputs a determination that no $3$-coloring exists.  Let 
$A=A_1 \cup A_2 \cup A_3$.  We analyze the 
structure of $G$ relative to $(A_1,A_2,A_3)$ and efficiently construct 
polynomially many subsets $D$ of $V(G)$ such that for each of them 
$G[A \cup D]$ only has a bounded number of  $3$-colorings, and {\em almost} 
all vertices of $V(G) \setminus (A \cup D)$ have  a neighbor in $D$.  Ignoring 
the {\em almost} qualification, we are now done using \ref{check} in 
polynomially many subproblems. In order to complete the proof, we guess a few 
more  vertices that need to be added to $D$ to create a dominating set in $G$, 
or show that certain subsets of $V(G)$ are monochromatic in all coloring of
$G$, which allows us to delete some vertices of $G$ without changing 
colorability. The last step is polynomially many applications of 
\ref{checkSubsets}.

\bigskip

\noindent This paper is organized as follows. In section 2 we prove \ref{tripod} and in section 3 we prove \ref{cleaning}, both of which are pre-processing procedures. In section 4 we prove \ref{Reduce}, which reduces the sizes of the lists
of all the vertices in  the graph except for a special stable set. In section 5 we prove a lemma,  \ref{lemma}, that we will use to deal with the vertices of 
this special stable set. In Section 6 we verify that \ref{lemma} can be applied in our situation. Finally, in Section 7 we put all the results together,
and show that we have reduced the problem to polynomially many subproblems,
each of which can be solved using \ref{checkSubsets}.

\section{Tripods}

In this section, we introduce a way to partition a graph that contains a triangle so that we begin to gain understanding into monochromatic sets this triangle forces. Additionally, we show that further simplifications are possible in the 
case that the graph we are considering is $P_7$-free. 

Let $(A_1,A_2,A_3)$ be a tripod in a graph $G$. We say $(A_1,A_2,A_3)$ is \textit{maximal} if there does not exist a vertex in $V(G)\setminus (A_1\cup A_2\cup A_3)$ which has a neighbor in two of $A_1,A_2,A_3$. 

\begin{lemma}\label{pod}

For any tripod $(A_1,A_2,A_3)$ in a graph $G$, for $\ell=1,2,3$ each $A_\ell$ is monochromatic with respect to any $3$-coloring of $G$. Moreover, no
color appears in two of $A_1, A_2, A_3$.

\end{lemma}

\begin{proof}

Let $A_1\cup A_2\cup A_3=\{a_1,...,a_m\}$ and $c$ be a $3$-coloring of $G$. We proceed by induction. Since $a_1-a_2-a_3-a_1$ is a triangle, it follows that $\{c(a_1),c(a_2),c(a_3)\}=\{1,2,3\}$. Suppose \ref{pod} holds for $\{a_1,...,a_{s-1}\}$, where $s>3$. Let $\{i,j,k\}=\{1,2,3\}$ so that $a_s\in A_i$. Since $(A_1,A_2,A_3)$ is a tripod, it follows that $a_s$ has a neighbor in $A_j \cap \{a_1,..,a_{s-1}\}$ and a neighbor in $A_k \cap \{a_1,..,a_{s-1}\}$. Inductively, it follows that every vertex in $A_j \cap \{a_1,..,a_{s-1}\}$ is assigned color $c(a_j)$ and that every vertex in $A_k \cap \{a_1,..,a_{s-1}\}$ is assigned color $c(a_k$). Since $c$ is a $3$-coloring, it follows that $c(a_s)=c(a_i)$. This proves \ref{pod}.

\end{proof}

We say a tripod $(A_1,A_2,A_3)$ is \textit{stable} if $A_i$ is stable for $i=1,2,3$. By \ref{pod}, it follows that if graph is $3$-colorable, then every tripod is stable. 

\begin{lemma}\label{tricomps}

If $(A_1,A_2,A_3)$ is a stable tripod in $G$, then $G[A_j\cup A_k]$ is a connected bipartite graph for all distinct $j,k\in \{1,2,3\}$.

\end{lemma}

\begin{proof}

Since $A_j$ and $A_k$ are stable, we only need to prove that $G[A_j\cup A_k]$ is connected. Suppose $A\cup B$ is a partition of $A_j\cup A_k$ such that both $A$ and $B$ are non-empty and $A$ is anticomplete to $B$. Since $a_1-a_2-a_3-a_1$ is a triangle, by symmetry, we may always assume $a_j,a_k\in A$. Choose $a_s\in B$ such that $s$ is minimal. It follows that $s>3$. By symmetry, we may assume $a_s\in A_j$. By definition, there exists $a_{s'}\in A_k\cap \{a_1,...,a_{s-1}\}$ adjacent to $a_s$. However, by minimality, $a_{s'}\in A$, contrary to $A$ being anticomplete to $B$. This proves \ref{tricomps}.

\end{proof}

We say a tripod $(A_1,A_2,A_3)$ in a graph $G$ is \textit{reducible} if for $\{i,j,k\}=\{1,2,3\}$ we have that $A_i$ is anticomplete to $V(G)\setminus (A_j\cup A_k)$. Suppose $(A_1,A_2,A_3)$ is a maximal reducible stable tripod in a graph $G$. By symmetry, we may assume that $A_1$ is anticomplete to $V(G)\setminus (A_2\cup A_3)$. Let $G_R$ be the graph obtained by deleting $A_1$ and contracting $A_2\cup A_3$ to an edge, that is, $V(G_R)= (V(G)\setminus (A_1\cup A_2\cup A_3)) \cup \{a_2',a_3'\}$ and

\begin{itemize}

\item $a'_2a'_3\in E(G_R)$,

\item $xy\in E(G_R)$ if and only if $xy\in E(G)$ for distinct $x,y\in V(G_R)\setminus \{a'_2,a'_3\}$,

\item $a'_2z\in E(G_R)$ if and only if $N_G(z)\cap A_2$ is non-empty where $z\in V(G_R)\setminus \{a'_2,a'_3\}$, and

\item $a'_3z\in E(G_R)$ if and only if $N_G(z)\cap A_3$ is non-empty where $z\in V(G_R)\setminus \{a'_2,a'_3\}$.

\end{itemize}

\noindent Note, $G_R$ can be constructed in time $O(|V(G)|^2)$. The following establishes the usefulness of the above reduction.
 
\begin{lemma} \label{G_R} 

Let $(A_1,A_2,A_3)$ be a maximal reducible stable tripod in a graph $G$ and assume that $A_1$ is anticomplete to $V(G)\setminus(A_2\cup A_3)$. Then the following hold:

\begin{enumerate}

\item If $G$ is a $P_7$-free graph, then $G_R$ is $P_7$-free.

\item If $G$ is connected, then $G_R$ is connected. 

\item $G$ is $3$-colorable if and only if $G_R$ is $3$-colorable, and specifically from a coloring of $G_R$ we can construct a coloring of $G$ in time $O(|V(G)|)$.

\end{enumerate}

\end{lemma}

\begin{proof}

First, we prove \ref{G_R}.1. Suppose $P$ is a copy of $P_7$ in $G_R$. Since $G$ is $P_7$-free, it follows that $V(P)\cap \{a'_2,a'_3\}$ is non-empty. First, suppose $|V(P)\cap\{a'_2,a'_3\}|=1$. By symmetry, we may assume $a'_2\in V(P)$. Since $G$ is $P_7$-free, it follows that $a'_2$ is an interior vertex of $P$, and so we can partition $P$ as $P'-p'-a'_2-p''-P''$, where $P',P''$ are paths, possibly empty. By construction, both $p', p''$ have a neighbor in $A_2$, and $V(P)$ is anticomplete to $A_1$. Since by \ref{tricomps} $G[A_1 \cup A_2]$ is connected, there exists a path $Q$ with ends $p'$ and $p''$ and interior in $A_1 \cup A_2$. But now
$P'-p'-Q-p''-P''$ is a  path in $G$ of length at least $7$, a contradiction. 

Thus, it follows that both $a'_2,a'_3\in V(P)$, and so we can partition $P$ as $S'-a'_2-a'_3-T'$, where $S',T'$ are paths, possibly empty. If $V(S)\neq \emptyset$, let $s'$ be the neighbor of $a'_2$ in $S'$; define $t'$ similarly.
Now  $s'$ has a neighbor in $A_2$, $t'$ has a neighbor in $A_3$, and 
$V(P) \setminus \{a_2,a_3,s',t'\}$ is anticomplete to $A_2 \cup A_3$. Since by
\ref{tricomps} $G[A_2 \cup A_3]$ is connected, it follows that there is a path
Q from $s'$ to $t'$ and with interior in $A_2 \cup A_3$. But now
$S'-s'-Q-t'-T'$ is a path in $G$ of length at least $7$, a 
contradiction. This proves \ref{G_R}.1.

Next we prove \ref{G_R}.2. Suppose $G_R$ is not connected, and let
$V(G_R)=X \cup Y$ such that $X,Y$ are non-empty and anticomplete
to each other. Since $a_2'$ is adjaent to $a_3'$, we may assume that 
$a_2',a_3' \in X$. Let 
$X'=(X \setminus \{a_2',a_3'\}) \cup (A_1 \cup A_2 \cup A_3)$. Then
$V(G)=X' \cup Y$, and $X',Y$ are anticomplete to each other, and so $G$ is not 
connected. This proves \ref{G_R}.2.

Finally, we prove \ref{G_R}.3. Suppose $c$ is a $3$-coloring of $G$. And so, we define the coloring $c'$ of $G$ as follows: For every $v\in V(G_R)$ set
$$c'(v) = \begin{cases}
c(a_2) & \text{,  \hspace{2ex}if $v=a'_2$} \\ 
c(a_3) & \text{,  \hspace{2ex}if $v=a'_3$} \\
c(v) & \text{,  \hspace{2ex}otherwise} \\
\end{cases}.$$
\noindent By construction, it clearly follows that $c'$ is a $3$-coloring of $G_R$.

Next, suppose $\hat{c}$ is a $3$-coloring of $G_R$. Since $a'_2$ is adjacent to $a'_3$, it follows that $\hat{c}(a'_2)\ne \hat{c}(a'_3)$. Take $\tilde{c}_1$ so that $\{\tilde{c}_1, \hat{c}(a'_2), \hat{c}(a'_3)\}=\{1,2,3\}$. Define the coloring $\tilde{c}$ of $G$ as follows: For every $v\in V(G)$ set

$$\tilde{c}(v) = \begin{cases}
\hspace{1ex}\tilde{c}_1 & \text{,  \hspace{2ex}if $v\in A_1$} \\ 
 \hat{c}(a'_2) & \text{,  \hspace{2ex}if $v\in A_2$} \\
 \hat{c}(a'_3) & \text{,  \hspace{2ex}if $v\in A_3$} \\
 \hat{c}(v) & \text{,  \hspace{2ex}otherwise} \\
\end{cases}.$$
\noindent By construction, it clearly follows that $\tilde{c}$ is a $3$-coloring of $G$ and the construction of $\tilde{c}$ takes $ O(|V(G)|) $. This proves \ref{G_R}.3.

\end{proof}

\noindent We say a tripod $(A_1,A_2,A_3)$ is \textit{normal} if it is stable, maximal and not reducible.

\begin{lemma}\label{tripod} There is an algorithm with the following specifications:

\bigskip

{\bf Input:} A connected graph $G$.

\bigskip

{\bf Output:}

\begin{enumerate}

\item a determination that $G$ is not $3$-colorable, or

\item a connected triangle-free graph $G'$ with $|V(G')|\leq |V(G)|$ such that $G'$ is $3$-colorable if and only if $G$ is $3$-colorable, or

\item a connected graph $G'$ with $|V(G')| \leq |V(G)|$ such that $G'$ is $3$-colorable if and only if $G$ is $3$-colorable, together with a normal tripod $(A_1,A_2,A_3)$ in $G'$.

\end{enumerate}

\bigskip

{\bf Running time:} $O(|V(G)|^3)$.

\bigskip

\noindent Additionally, any $3$-coloring of $G'$ can be extended to a $3$-coloring of $G$ in time $ O(|V(G)^2|) $.

\end{lemma}

\begin{proof}

In time $O(|V(G)|^3)$, we can determine if $G$ is triangle-free. If so return the triangle-free graph $G'=G$ and halt. Otherwise, we may assume there exist $a_1,a_2,a_3\in V(G)$ such that $a_1-a_2-a_3-a_1$ is a triangle. Next, we try and grow this triangle into a normal tripod. Initialize $A_i=\{a_i\}$ for $i=1,2,3$. Assume $A_1\cup A_2\cup A_3=\{a_1,...,a_m\}$ and consider $v\in V(G)\setminus (A_1\cup A_2\cup A_3)$ such that $v$ has a neighbor in $A_i$ and $A_j$ for $\{i,j,k\}=\{1,2,3\}$. If $v$ is anticomplete to $A_k$, then set $a_{m+1}=v$ and $ A_k=A_k\cup \{a_{m+1}\}$. If $v$ has a neighbor in $A_k$, then, by \ref{pod}, we may return that $G$ is not $3$-colorable and halt. Repeat this procedure again until either we determine that $G$ is not $3$-colorable or there does not exists any $v\in V(G)\setminus (A_1\cup A_2\cup A_3)$ such that $v$ has a neighbor in $A_i$ and $A_j$ for $\{i,j,k\}=\{1,2,3\}$. By construction, this procedure either halts or yields a maximal, stable tripod $(A_1,A_2,A_3)$ from the triangle $a_1-a_2-a_3-a_1$. In time $O(|V(G)|^2)$, we can verify if $A_i$ is anticomplete to $V(G)\setminus (A_j\cup A_k)$ for some $i\in \{1,2,3\}$, that is, if $(A_1,A_2,A_3)$ is reducible. If not, then return the normal tripod $(A_1,A_2,A_3)$ for $G'=G$ and halt. Otherwise, by symmetry, we may assume $A_1$ is anticomplete to $V(G)\setminus (A_2\cup A_3)$. By \ref{G_R}.2, it follows that $G_R$ is connected, 
$G$ is $3$-colorable if and only if $G_{R}$ is $3$-colorable, and a $ 3 $-coloring of $G_{R}$ can be extended to a $ 3 $-coloring of $G$ in time $|V(G)|$. Now, repeat the steps described above with $G_R$.  This procedure can be carried out in time $O(|V(G)|^3)$.  This proves \ref{tripod}.

\end{proof}

\section{Cleaning}

In this section, we identify a configuration that, if present in $G$,  allows 
us to efficiently find a graph $G'$ with $|V(G')|< |V(G)|$ which is 
$3$-colorable if and only $G$ is $3$-colorable. 

Let $G$ be a graph, and let  $(A_1,A_2,A_3)$ be a tripod in $G$.
We say $v\in V(G)$ is a \textit{connected vertex} if 
$G[N_G(v)]$ is connected. We say that a graph is {\em $(A_1,A_2,A_3)$-clean} if 
every  connected vertex in $V(G) \setminus (A_1 \cup A_2 \cup A_3)$ has a 
neighbor in $A_1 \cup A_2 \cup A_3$, and $(A_1,A_2,A_3)$ is a normal
tripod in $G$.

\begin{lemma}\label{cleaning} Let $(A_1,A_2,A_3)$ be a normal tripod in $G$.
There is an algorithm with the following specifications:

\bigskip

{\bf Input:} A connected $P_7$-free graph $G$.

\bigskip

{\bf Output:} A connected $(A_1,A_2,A_3)$-clean  $P_7$-free graph $G'$ with $|V(G')|\leq |V(G)|$ such that $G'$ is $3$-colorable if and only if $G$ is $3$-colorable, or a determination that $G$ is not $3$-colorable. 

\bigskip

{\bf Running time:} $O(|V(G)|^4)$.

\bigskip

\noindent Additionally, any $3$-coloring of $G'$ can be extended to a 
$3$-coloring of $G$ in time $O(|V(G)|^2)$.

\end{lemma}

\begin{proof} 

First, for every $v \in V(G)$, check if  $G[N(v)]$ is $2$-colorable.
This can be done in time $O(|V(G)|^3)$, and if the answer is ``no'' for
some $v$, we can stop and output that $G$ is not $3$-colorable.

Thus we may assume that $G[N(v)]$ is $2$-colorable for every $v \in V(G)$.
Let $Y$ be the set of vertices of $V(G) \setminus (A_1 \cup A_2 \cup A_3)$ 
that are anticomplete to $A_1 \cup A_2 \cup A_3$.
In time $O(|V(G)|^3)$, we can find a connected vertex in  $Y$
or determine that none exists.  If no vertex in $Y$ is connected, output 
$G'=G$.

Suppose $v \in Y$ is connected. Define $G_v$ as follows. If 
$|N(v)|=1$, let $G_v=G\setminus v$. Otherwise, let $(A,B)$ be the unique 
bipartition of $G[N(v)]$. It follows that $\{v\}\cup A\cup B$ is a maximal 
reducible  stable tripod in $G$. Let $G_v$ be the graph obtained from $G$ by 
deleting $v$ and contracting $N_G(v)$ to an edge, that is, $G_R$ with respect 
to  $\{v\}\cup A\cup B$. Now, by~\ref{G_R}, it follows that $G_v$ is 
connected, and that $G_v$ is 
colorable if and only if $G$ is colorable. Moreover, since $v \in Y$,
it follows that $(A_1,A_2,A_3)$ is a normal tripod in $G_v$.

Now recursively applying the procedure to $G_v$, \ref{cleaning} follows.
\end{proof}

Given a graph $G$, we say that $X \subseteq V(G)$ is a {\em a homogeneous set in $G$} if $X \neq V(G)$, and every vertex of $V(G) \setminus X$ is either
complete or anticomplete to $X$.  We end the section with the following lemma.

\begin{lemma}\label{hset} Let $X$ be a  homogeneous set in a connected graph $G$ such that $G[X]$ is connected, $X\ne V(G)$ and $|X|>1$. Then $X$ contains a connected vertex.

\end{lemma}

\begin{proof}
Consider $v\in X$ and define $X'=N(v)\cap X$ and $Y=N(v)\cap (V(G)\setminus X)$. Since $G$ is connected, it follows that $V(G)\setminus X$ is not anticomplete to $X$, and so $Y$ is non-empty. Since $G[X]$ is connected and $|X|>1$, it 
follows  that $X'$ is non-empty. Since $X$ is a homogeneous set, it follows that $Y$ is complete to $X'$, implying that $G[N(v)]$ is connected. This proves 
\ref{hset}.

\end{proof}

\section{Reducing the Graph}

The main result of this section is~\ref{Reduce}.  It allows us (at the 
expense of branching into polynomially many subproblems) to reduce the lists of some 
of the vertices of  the graph to size two, and get some control over the 
remaining vertices. More precisely, \ref{Reduce} reduces the problem to the
case when the set of vertices whose list has size three is stable, and the 
neighbors of every such vertex satisfy certain technical conditions.  These
conditions are designed with the goal of using \ref{lemma}. In \ref{MainLemma} 
we verify that the conclusion of~\ref{Reduce} is in fact sufficient for 
applying~\ref{lemma}.

For a fixed subset $X$ of $V(G)$, we say that a vertex \textit{$v\in V(G)\setminus X$ is mixed on an edge of $X$}, if there exist adjacent $x,y\in X$ such that $v$ is adjacent to $x$ and non-adjacent to $y$. Similarly,  we say a vertex \textit{$v\in V(G)\setminus X$ is mixed on a non-edge of $X$}, if there exist non-adjacent $x,y\in X$ such that $v$ is adjacent to $x$ and non-adjacent to $y$.

\begin{lemma}\label{Reduce}
Let $A=(A_1,A_2,A_3)$ be a normal tripod in a connected, $(A_1,A_2,A_3)$-clean $P_7$-free graph and partition $V(G)=A\cup X \cup Y \cup Z$, such that 
\begin{itemize}

\item $A=A_1\cup A_2\cup A_3$,

\item $X$ is  the set of vertices of $V(G)\setminus A$ with a neighbor in $A$,

\item $Y$ is  the set of vertices of $V(G)\setminus (A\cup X)$ with a neighbor in $X$, 

\item $Z=V(G) \setminus (A\cup X\cup Y)$.

\end{itemize}
For $i=1,2,3$, let $X_i$ be the set of vertices of $V(G)\setminus A$ with a 
neighbor in $A_i$.

There exists a set of $O(|V(G)|^{12})$ palettes $\mathcal{L}$ of $G$ such that

\bigskip

\noindent \textit{(a)} Each $L\in\mathcal{L}$ has order $3$ and $ |L(v)| \leq 2 $ for every $v \in A\cup X$, and

\bigskip

\noindent \textit{(b)} $G$ has a $3$-coloring if and only if $(G,\mathcal{L})$ is colorable.

\bigskip

\noindent Moreover, $\mathcal{L}$ can be computed in time $O(|V(G)|^{15})$.

\bigskip

\noindent  For each $L\in\mathcal{L}$, let $P_L$ be the set of 
vertices $y \in Y\cup Z$ with $|L(y)|=3$. Then the following hold:

\bigskip

\noindent \textit{(c)} $P_L$ is stable.

\bigskip

\noindent \textit{(d)} There exist subsets $X'\subset X$, $Y_0\subset Y$, and vertices $s_\ell\in X_\ell \cap X'$ for $\ell=1,2,3$, such that 
\begin{itemize}
\item $|L(x)|=1$ for all $x\in X'\cup Y_0$, and 
\item $Y_0$ is complete to $\{s_1,s_2,s_3\}$, and
\item letting $Y'$ be the set of vertices in $Y\cup Z$ with a neighbor in 
$X'\cup Y_0$, we have that $P_L$ is anticomplete to $(Y\cup Z) \setminus Y'$.
\end{itemize}
\bigskip

\noindent Additionally, for every $(i,j,k)=(1,2,3)$ and  $L\in \mathcal{L}$ the
 following hold:

\bigskip

\noindent \textit{(e)} If $v\in Y'$ with $L(v)=\{i,j\}$, then there exists $u\in N(v)\cap (X'\cup Y_0)$ with $L(u)=\{k\}$,

\bigskip

\noindent \textit{(f)} If $v\in X'\cap X_j$ with $L(v)=\{i\}$, then either there exists $u\in N(v)$ such that $L(u)=\{k\}$, or every $y \in Y$ with a neighbor in $X_j$ has $L(y)=\{j\}$.

\bigskip

\noindent \textit{(g)} If $v\in Y_0$ with $L(v)=\{i\}$, then there exist $u,w\in N(v) \cap \{s_1,s_2,s_3\}$ such that $L(u)=\{k\}$ and $L(w)=\{j\}$. 

\end{lemma}

\begin{proof} Since a normal tripod is maximal and not reducible, it follows 
that
\begin{itemize}
\item $X_\ell$ is non-empty for $\ell=1,2,3$.

\item $X_i\cap X_j=\emptyset$ for all distinct $i,j\in\{1,2,3\}$.

\item $X_1\cup X_2\cup X_3=X$.

\end{itemize}
Let $\ell \in \{1,2,3\}$. Let $\mathcal{S}_\ell$ be the set of all quadruples 
$S=(P,Q_1,Q_2,Q_3)$ such that

\begin{itemize}
\item $P=\{p\}$ and $p \in X_\ell$.
\item For $i,j \in \{1,2,3\}$, if $Q_i \neq \emptyset$ and $j<i$, then $Q_j \neq \emptyset$.
\item either $Q_1=\emptyset$, or $Q_1=\{q_1\}$, $q_1  \in Y$ and $p$ is 
adjacent to $q_1$.
\item either $Q_2=\emptyset$, or $Q_2=\{q_2\}$, $q_2 \in Y \cup Z$ and
$q_2$ is adjacent to $q_1$ and not to $p$.
\item either $Q_3=\emptyset$, or $Q_3=\{q_3\}$, $q_3 \in Y$, and $q_3$ is
adjacent to $p$ and anticomplete to $\{q_1,q_2\}$.
\end{itemize}
Let $E(S)=P \cup Q_1 \cup Q_2 \cup Q_3$. We write $P(S)=P$, and $Q_i(S)=Q_i$
for $i=1,2,3$.
Let $\mathcal{S}=\{(S_1,S_2,S_3)$ such that $ S_\ell \in \mathcal{S}_\ell\}$.
Then $|\mathcal{S}|=O(|V(G)|^{12})$.

Let us say that $y \in Y$ is an \textit{$i$-cap} if there exist $x \in X_i$
and $y' \in (Y \cup Z) \setminus \{y\}$ such that $x-y-y'$ is a  path. 
Initialize the palette $L$:

$$L(v) = \begin{cases}
 \{1\} & \text{,\hspace{3ex} if $v\in A_1$} \\
 \{2\} & \text{,\hspace{3ex} if $v\in A_2$} \\
 \{3\} & \text{,\hspace{3ex} if $v\in A_3$} \\
 \{2,3\} & \text{,\hspace{3ex} if $v\in X_1$} \\
 \{1,3\} & \text{,\hspace{3ex} if $v\in X_2$} \\
 \{1,2\} & \text{,\hspace{3ex} if $v\in X_3$} \\
 \{1,2,3\} & \text{,\hspace{3ex} otherwise} \\
\end{cases}$$

\noindent Clearly, by renaming the colors, $G$ has a $3$-coloring if and only if $(G,L)$ is colorable.  The sets $(S_1,S_2,S_3)$ are designed to ``guess'' 
information about certain types of colorings of $G$ (type I--IV colorings 
defined later). Next we ``trim'' the collection $\mathcal{S}$, with the goal to 
only keep the sets that record legal colorings of each type.

For every $S=(S_1,S_2,S_3)$, proceed as follows.
If $Q_3(S_i)=\emptyset$ and $Q_2(S_i) \neq \emptyset$, let $M(S_i)$ be the set 
of vertices of  $Y$ that are  complete to  $P(S_i)$ and anticomplete to 
$Q_1(S_i) \cup Q_2(S_i)$, otherwise let $M(S_i)=\emptyset$.
If $Q_2(S_i)=Q_3(S_i)=\emptyset$, let $H(S_i)$ be the set of all $i$-caps,
and otherwise let $H(S_i)=\emptyset$. 
If  for some  $i \neq j \in \{1,2,3\}$
$Q_1(S_i)=Q_1(S_j)=\emptyset$ and there is $y \in Y$ with both a neighbor in 
$X_i$ and $X_j$, discard $S$.

Next suppose that  for some $i \in \{1,2,3\}$, $Q_3(S_i)=\emptyset$, 
and $Q_1(S_i), Q_2(S_i)  \neq \emptyset$. If
there exist $x \in X_i$ and $y_1,y_2 \in Y$ such that 
\begin{itemize}
\item $x$ is adjacent to $y_1$ and not to $y_2$
\item $y_1$ is adjacent to $y_2$,
\item $M(S_i) \cup Q_2(S_i)$ is anticomplete to $\{y_1,y_2\}$, and 
\item $x$ is complete to $M(S_i) \cup Q_2(S_i)$
\end{itemize}
then discard $S$.

Otherwise, let  
$E(S)=\bigcup_{i \in \{1,2,3\}} (E(S_i) \cup M(S_i) \cup H(S_i))$, and 
let $c$ be a coloring of  $G[E(S)]$ such that $c(v)=i$ for every
$v \in M(S_i) \cup H(S_i)$.  If for some $i \in \{1,2,3\}$,
$Q_1(S_i) \neq \emptyset$ and the vertex of $Q_1(S_i)$ is colored $i$,
discard $c$.  If for some $i$, $Q_3(S_i) \neq \emptyset$ and
the vertex of $Q_3(S_i)$ is colored $i$, discard $c$.

Otherwise, define the subpalette $L^S_c$  of $L$ as 
follows:

$${L}^S_c(v) = \begin{cases}
 c(v) & \text{, \hspace{2ex}if $v \in E(S)$} \\
1 & \text{, \hspace{2ex} if $Q_1(S_1)=\emptyset$, and $v \in Y$ and $v$ has a 
neighbor in $X_1$}\\
2 & \text{, \hspace{2ex} if $Q_1(S_2)=\emptyset$, and $v \in Y$ and $v$ has a 
neighbor in $X_2$}\\
 3 & \text{, \hspace{2ex} if $Q_1(S_3)=\emptyset$, and $v \in Y$ and $v$ has a 
neighbor in $X_3$}\\
 L(v) & \text{, \hspace{2ex}otherwise} \\
 \end{cases} $$

Fix $S=(S_1,S_2,S_3)$. Let $\ell \in \{1,2,3\}$ and
let $X_\ell'^S$ be the set of vertices $x \in X_\ell$ with a neighbor $w$ in
$E(S_\ell)$ such that $c(w) \neq \ell$, and let 

$$X'_S=X_1'^S \cup X_2'^S \cup X_3'^S \cup (E(S) \cap X)).$$

Let $Y_0^S$ be the set of vertices of $Y$ that are complete to 
$P(S_1) \cup P(S_2) \cup P(S_3)$.

Let $Y'_S$ be the set of vertices of $(Y \cup Z) \setminus (Y^s_0 \cup E(S))$ 
with a neighbor in  $Y_0^S \cup X'_S$.
We now carry out three rounds of  updating: first,
for every $\ell \in \{1,2,3\}$,  update $X_\ell'$ with respect to 
$E(S_\ell)$, then update $Y'_S$ with respect to 
$Y_0^S \cup X'_S$ and finally update $Y \setminus Y'_S$ with respect to  
$Y_0^S \cup X'_S \cup E(S)$. This takes time   $O(|V(G)|^2)$.
Let $\mathcal{L}$ be the set of all the palettes
$L_c^S$ thus generated. Then $|\mathcal{L}|=O(|V(G)|^{12})$, and
$\mathcal{L}$ can be constructed in time $O(|V(G)|^{15})$.
Clearly, \textit{(a)} holds.

We now define four different types of colorings of $G$ that are
needed to prove  \textit{(b)}.
Let $c$ be a coloring of $G$ and let $\ell \in \{1,2,3\}$. 
We say that  \textit{$c$ is a 
type I coloring with respect to $\ell$} if there exist vertices
$(p,q_1,q_2)$, where the following hold:

\begin{itemize}
\item $p \in X_\ell$
\item $q_1,q_2 \in Y \cup Z$ such that $p$ is adjacent to $q_1$ and not to $q_2$, 
and $q_1$ is adjacent to $q_2$
\item $c(q_1) \neq \ell$, and $c(q_2) \neq \ell$.
\end{itemize}

We say that   \textit{$c$ is a 
type II coloring with respect to $\ell$} if $c$ is not a 
type I coloring with respect to $\ell$ and there exist vertices
$(p,q_1,q_2,q_3)$, where the following hold:

\begin{itemize}
\item $p \in X_\ell$
\item $q_1,q_2 \in Y \cup Z$ such that $p$ is adjacent to $q_1$ and not to 
$q_2$,  and $q_1$ is adjacent to $q_2$
\item $q_3 \in Y$, and $q_3$ is adjacent to $p$ and anticomplete to 
$\{q_1,q_2\}$. 
\item $c(q_1) \neq \ell$, $c(q_2)=\ell$ and $c(q_3) \neq \ell$.
\end{itemize}

We say that \textit{$c$ is a 
type III coloring with respect to $\ell$} if $c$ is not a 
type I or type II coloring with respect to $\ell$ and there exist vertices
$(p,q_1,q_2)$, where the following hold:

\begin{itemize}
\item $p \in X_\ell$
\item $q_1,q_2\in Y \cup Z$ such that $p$ is adjacent to $q_1$ and not to
$q_2$, and $q_1$ is adjacent to $q_2$. 
\item $c(q_1) \neq \ell$, and $c(q_2)=\ell$.
\end{itemize}

We say that \textit{$c$ is a 
type IV coloring with respect to $\ell$} if $c$ is not a type I, type II, or type III coloring with respect to $\ell$ and there exist vertices
$(p,q_1)$, where the following hold:

\begin{itemize}
\item $p \in X_\ell$
\item $q_1 \in Y$ such that $p$ is adjacent to $q_1$.
\item $q_1$ is not an $\ell$-cap.
\item $c(q_1) \neq \ell$.
\item if $y$ is an $\ell$-cap, then $c(y)=\ell$.
\end{itemize}

We claim that if $c$ is a coloring of $G$ that is not of type I,II,III or IV
for some $i \in \{1,2,3\}$, then $c(y)=i$ for every $y \in Y$ with a 
neighbor in $X_i$. For suppose $c(y) \neq i$ for some $y \in Y$ with a 
neighbor $x \in X_i$. If $y$ can be chosen to be an $i$-cap, then $c$ is a type
I,II or III coloring, and otherwise $c$ is a type IV coloring. This proves
the claim.

Next we prove \textit{(b)}.
Clearly if $c$ is a coloring of $(G,L)$ for some $L \in \mathcal{L}$,
then $c$ is a coloring of $G$. We show that if $G$ is colorable, then 
$(G,L)$ is colorable  for some $L \in \mathcal{L}$. 
 
Let $c$ be a coloring of $G$.  Suppose first that  $c$ is a type I,II, III or 
IV coloring with  respect to $1$. Then there exist $p$ and possibly 
$q_1,q_2,q_3$ as in the 
definition of a type I,II, III or IV coloring.  If $c$ is a type III coloring,
let $M_1$ be the set of all vertices in $Y$ that are adjacent to $p$ and 
anticomplete to $\{q_1,q_2\}$. If $c$ is a type IV coloring, let $H_1$ be the 
set of all $1$-caps. Moreover, if $c$ is a type III coloring,  we may assume 
that  $p,q_1,q_2$ are chosen in such a way that $M_1$ is maximal, and so there do
not  exist  $x \in X_i$ and $y_1,y_2  \in Y \cup Z$ such that 
\begin{itemize}
\item $x$ is adjacent to $y_1$ and not to $y_2$
\item $y_1$ is adjacent to $y_2$
\item $M_1 \cup \{q_2\}$ is anticomplete to $\{y_1,y_2\}$
\item $x$ is complete to $M_1 \cup \{q_2\}$.
\end{itemize}
Also, if $c$ is a type IV coloring of $G$, then $c(y)=1$ for every $y \in H_1$.
Let $S_1=(P,Q_1,Q_2,Q_3)$ such that $P=\{p\}$, and for $i \in \{1,2,3\}$ 
either $Q_i=\{q_i\}$ or $Q_i = \emptyset$ if $q_i$ is not defined. 

If $c$ is  not a type I, II, III  or IV coloring with respect to $1$, choose 
$p \in X_1$ and set  $S_1=(\{p\}, \emptyset, \emptyset,\emptyset)$. 

Define $S_2,M_2,H_2$ and $S_3,M_3,H_3$ similarly, and let $S=(S_1,S_2,S_3)$. 
Recall that $E(S)=\bigcup_{i \in  \{1,2,3\}}(E(S_i) \cup M_i \cup H_i)$.
Now let $d$ be the restriction of $c$ to $G[E(S)]$.
It is easy to see that $c(v) \in L^S_d(v)$ for every $v \in V(G)$, and so
$c$ is a coloring of $(G,L_d^S)$. Thus \textit{(b)} holds.

Fix $S \in \mathcal{S}$, $c$ a coloring of $E(S)$ as described at the start of 
the proof,  and $L^S_c \in \mathcal{L}$. For $i \in \{1,2,3\}$,
let $P(S_i)=\{p_i\}$. 
Let $X''_S=X \setminus X'_S$. 
Let $T_S=(Y \cup Z) \setminus (Y'_S \cup Y_0^S \cup E(S))$. 
Now  $|L^S_c(v)|=1$ for every $v \in X'_S \cup E(S)$. Since at least two colors 
appear in $P(S_1) \cup P(S_2) \cup P(S_3)$, it follows that
$|L^S_c(v)|=1$ for every $v \in Y^s_0$. Setting
$s_i=p_i$,  we observe that \textit{(g)} holds.
Consequently, since every vertex of $Y'_S$ has a neighbor in
$X'_S \cup Y_0^S$, it follows that  $|L^S_c(v)| \leq 2$ for every 
$Y'_S$ and  \textit{(e)} holds.
Next we show that $\textit{(f)}$ holds. Let $i \neq j \in \{1,2,3\}$, 
let  $v \in X'_S \cap X_j$ and suppose that $L^S_c(v)=\{i\}$. If $v \in X_j'$,
then  $L^S_c(v)$ was changed in the first round of updating, and the 
assertion of $\textit{(f)}$ holds. Thus we may assume that $v \in P(S_j)$,
and $Q_1(S_i)=\emptyset$. But then every $y \in Y$ with a neighbor in $X_j$ 
has $L_c^S(y)=\{j\}$, and again $\textit{(f)}$ holds.

Next we prove a few structural statements about $G$, that will allow us 
to prove \textit{(c)} and \textit{(d)}. 

\bigskip

\noindent \textit{(1) If $x \in X_i$ and $y_1,y_2,y_3 \in Y \cup Z$ are such that $x-y_1-y_2-y_3$ is a path, then every vertex of $X_j \cup X_k$ has a neighbor in $\{y_1,y_2,y_3\} $}.

\bigskip

\noindent Proof: Suppose not. By symmetry, we may assume there exists a vertex $v\in X_j$ anticomplete to $\{y_1,y_2,y_3\}$. Suppose first that $v$ is non-adjacent to $x$. Since by \ref{tricomps} 
$G[A_i \cup A_j]$ is connected, and since both $x$ and $v$ have neighbors
in $A_i \cup A_j$, it follows that there exists a path $P$ from $x$
to $v$ with interior in $A_i \cup A_j$.
It follows that $V(P)$ is anticomplete to $\{y_1,y_2,y_3\}$ and so $v-P-x-y_1-y_2-y_3$ contains a $P_7$, a contradiction. Thus $v$ is adjacent to $x$. Let  $ a\in N(v)\cap A_j $ and $  b\in N(a)\cap A_k $, then $b-a-v-x-y_1-y_2-y_3$ is a  $P_7$ in $G$, a contradiction. This proves \textit{(1)}.

\bigskip

\noindent \textit{(2) If $x \in X_i$, $z \in Y$, and $y_1,y_2\in Y''_S$ are such that $x-z-y_1-y_2$ is a path, then $z\in Y^S_0 $.}

\bigskip

\noindent Proof:  We may assume that $i=1$. By \textit{(1)}, each 
of $p_2,p_3$ has a neighbor in $\{y_1,y_2,z\}$. Since $y_1,y_2\in Y''_S $, it follows that $\{y_1,y_2\}$ is anticomplete to $\{p_2,p_3\}$. This implies that $z$ is complete to $\{p_2,p_3\}$, and so $v-z-y_1-y_2$ is a path for every $v \in \{p_2,p_3\}$. Now, by the same argument it follows that $z$ is adjacent to $p_1$. Hence, $z\in Y^s_0$. This proves \textit{(2)}.

\bigskip

Let $P_L$ be the set of vertices  $t \in T_S$ with $|L_c^S(t)|=3$. 
From the definition of $L_c^S$, it follows that if $v \in T_S \setminus P_L$, 
then  for some $i \in \{1,2,3\}$, $v$ has a neighbor in 
$E(S)=(E(S_i) \setminus X) \cup H(S_i) \cup M(S_i)$.

\bigskip

\noindent \textit{(3) No vertex of $X''_S$ is mixed on an edge of $P_L$.}

\bigskip

\noindent Proof: Suppose $x-y_1-y_2$ is a
path,  where $x \in X_1 \cap X''_S$, and $y_1,y_2 \in P_L$. Then 
$y_1$ is an $i$-cap and  $L_c^S(y_1) \neq \{1\}$.  It follows that 
$Q_1(S_1) \neq \emptyset$, $Q_2(S_1) \neq \emptyset$, and 
$c(Q_1(S_1)) \neq 1$. Write $p=p_1$. Let $Q_1(S_1)=\{q_1\}$.
Since $x \in X''_S \cap X_\ell$, it follows that $x$ is anticomplete
to $P(S_1) \cup Q_1(S_1)$. Since $y_1,y_2 \in P_L$, it follows that 
$\{y_1,y_2\}$ is anticomplete to $E(S_1)$.

Let $Q_2(S_1)=\{q_2\}$.
Suppose first that $x$ is non-adjacent to $q_2$. Let $P$ be a path 
from $x$ to $p$ with interior in $A_1 \cup A_2$ (such a path exists 
by \ref{tricomps}).  Now $y_2-y_1-x-P-p-q_1-q_2$ is a  path with at 
least  seven vertices, a contradiction.  This proves that $x$ is adjacent to 
$q_2$, and since $x \in X''_S$, we deduce that $c(q_2)=1$.

Next suppose $Q_3(S_1) \neq \emptyset$; let $Q_3(S_1)=\{q_3\}$. Then
$c(q_3) \neq 1$, and so $x$ is non-adjacent to $q_3$. Now
$y_2-y_1-x-q_2-q_1-p-q_3$  is a $P_7$, a contradiction. This proves that
 $Q_3(S_1)=\emptyset$. Recall  that when $Q_2(S_1) \neq \emptyset$ and
$Q_3(S_1) \neq \emptyset$, $M(S_1)$ is defined to be the  set of all vertices 
of  $Y$ that are adjacent to $p$ and anticomplete to $\{q_1,q_2\}$. Then
$L_c^S(v)=1$ for every $m \in M(S_1)$, and $\{y_1,y_2\}$ is anticomplete
to $M(S_1)$. Consequently,  since $y_2-y_1-x-q_2-q_1-p-m$ is not a $P_7$ for 
any $m \in M(S_1)$, we deduce that $x$ is complete to $M(S_1)$, and thus
the quadruple $S_1$ was discarded during the construction of $\mathcal{L}$,
a contradiction.
 This proves 
\textit{(3)}.

\bigskip

\noindent \textit{(4)  No vertex of $Y'_S$ is mixed on an edge of $P_L$.}

\bigskip

\noindent Proof: Suppose $y-y_1-y_2$ is a path, where $y \in Y'_S$,
and $y_1,y_2 \in P_L$. Then $y \not \in E(S) \cup Y_0^S$. If $y$
has a neighbor in $x \in X'_S$, then $x-y-y_1-y_2$ is 
 path, and so by \textit{(2)} $y \in Y_0^S$, a contradiction. 
This proves that $y$ is anticomplete to $X'_S$, and so
$y$ has a neighbor in $y_0 \in Y_0^S$. By the definition of $Y_0^S$,
$y_0$ is adjacent to $p_1$. Let $a_1 \in A_1$ be
adjacent to $p_1$, and let $a_2 \in A_2$ be adjacent to $a_1$.
Now $a_2-a_1-p_1-y_0-y-y_1-y_2$ is  $P_7$ in $G$, a contradiction. This proves
\textit{(4)}.

\bigskip

\noindent \textit{(5) If $T_S \setminus P_L$ is anticomplete to $P_L$.}

\bigskip

\noindent Proof: Suppose $t \in T_S \setminus P_L$ has a neighbor $p \in P_L$.
Then $t$ has a neighbor $w \in E(S) \setminus X$, and since $|L_c^S(p)|=3$, 
it follows that $w$ is non-adjacent to  $p$. Suppose first that
$w \in Q_1(S_i) \cup Q_3(S_i) \cup M(S_i)$ for some $i \in \{1,2,3\}$. Then $w$ 
has a neighbor $x \in X'$, and so
$x-w-t-p$ is a path. Now \textit{(2)} implies that $w \in Y_0$,
and therefore $t \in Y'$, contrary to the fact that $t \in T_S$.

Next suppose that $w \in Q_2(S_i)$ for some $i \in \{1,2,3\}$. We may assume
$i=1$.  Let  $Q_1(S_1)=q_1$. Let $a \in A_1$ be a neighbor of $p_1$,
and let $a' \in A_2$ be adjacent to $a$. Then $a'-a-p_1-q_1-w-t-p$ is a
$P_7$ in $G$, a contradiction. 

Consequently,  $w \in H(S_i)$ for some $i \in \{1,2,3\}$. In particular,
$H(S_i) \neq \emptyset$,  and so 
$L_c^S(h)=\{i\}$ for every $i$-cap. Let $x \in X_i$ be adjacent to $w$. 
If $x$ is anticomplete to 
$\{t,p\}$, then again by \textit{(2)} $w \in Y_0^S$, a contradiction.
So, since $t,p \not \in H(S_i)$  it follows that $x$ is complete
to $\{t,p\}$, and in particular $p \in Y$. Therefore 
$N(p) \cap X \neq \emptyset$. Moreover, the fact that 
$p \not \in H(S_i)$ implies that $N(p) \cap X$ is complete to $N(p) \setminus X$.
Since $t \in N(p) \setminus X$, it follows that $p$ is a connected vertex,  
contrary to the fact that $G$ is $(A_1,A_2,A_3)$-clean. This proves \textit{(5)}.

\bigskip

Now by \textit{(3)}, \textit{(4)} and \textit{(5)}, for  every connected component $C$ of 
$P_L$, $V(C)$  is a homogeneous
set. Since no vertex of $P_L$ is connected, by~\ref{hset} $|V(C)|=1$, $P_L$ is stable and \textit{(c)} holds. Finally, setting $s_i=p_i$ for $i \in \{1,2,3\}$, $X'=X'_S$, and $Y_0=Y'_S$, 
\textit{(5)} implies that  \textit{(d)} holds.
This completes the proof of~\ref{Reduce}. \end{proof}

\section{A Lemma}

This section contains  a lemma that captures the properties of the set $P_L$
from \ref{Reduce} that makes it possible to reduce the size of the  lists of 
the vertices in this set.

\begin{lemma} \label{lemma}

Let $L$ be an order 3 palette of a connected $P_7$-free graph $G$. Let $Z$ be a set of subsets of $V(G)$. Suppose there exists disjoint non-empty  subsets 
$S_1, S_2 ,S_3$ of $V(G)$ satisfying the following:

\begin{itemize}

\item $L(v)=\{1,2,3\}\setminus \{\ell\}$ for every $v\in S_\ell$ where $\ell\in\{1,2,3\}$.

\item Let $i,j\in\{1,2,3\}$, and let $u_i, v_i \in S_i$ and $u_j,v_j \in S_j$,
such that $\{u_i,v_i,u_j,v_j\}$ is a stable set. Then there exists a path
$P$ with ends $a,b \in \{u_i,v_j,u_j,v_j\}$ such that
\begin{enumerate}
\item $\{a,b\} \neq \{u_i,u_j\}$ and $\{a,b\} \neq \{v_i,v_j\}$,  
\item $|L(w)|=1$ for every  interior vertex $w$ of $P$, and 
\item $V(P) \setminus \{a,b\}$ is disjoint from and anticomplete to 
$\{u_i,v_j,u_j,v_j\} \setminus \{a,b\}$.
\end{enumerate}

\item For every distinct pair $i,j\in\{1,2,3\}$ and $u\in S_i$ there exist vertices $v$ and $w$, such that $u-v-w$ is a path where both $v$ and $w$ are anticomplete to $S_j$ with $|L(v)|=|L(w)|=1$. 
\end{itemize}

\noindent Given a vertex $x\in V(G)$, define $N_\ell(x)=N(x)\cap S_\ell$ for $\ell=1,2,3$.
Let $X \subset V(G)$ be such that $N(x) \subseteq S_1 \cup S_2 \cup S_3$
for every $x \in X$, and no vertex of $X$ is connected.

\noindent Then there exists a set   $\mathcal{P}$ of $O(|V(G)|^9)$
restrictions of $(G,L,Z)$ such that the following hold:
\bigskip

\noindent \textit{(a)} For every $(G',L',Z') \in \mathcal{P}$,
$|L'(v)| \leq 2$ for every $v \in X \cap V(G')$, and $|Z'|=O(|V(G)|+|Z|)$, and

\bigskip
  
\noindent \textit{(b)} $(G,L,Z)$ is colorable if and only if $\mathcal{P}$ 
is colorable.

\bigskip

\noindent \textit{Moreover, $\mathcal{L}$ can be constructed in time 
$O(|V(G)|^{10})$, and a $3$-coloring of a restriction in $ \mathcal{P} $ can be extended to a $3$-coloring of $G$ in $ O(|V(G)|^2) $.}

\noindent

\bigskip

\end{lemma}

\begin{proof}

Let $X'$ be the set of vertices $x \in X$ with $|L(x)|=3$. 
If $X'=\emptyset$, let $\mathcal{P}=\{(G,L,Z)\}$.

By updating, we may assume that  for every  $x \in X'$ and $y$ adjacent to 
$x$, $|L(y)|\geq 2$. If $N(x) \subseteq S_i$ for some $x \in X'$ and 
$i \in \{1,2,3\}$, then setting $L(x)=\{i\}$ does not change the colorability
of $(G,L,Z)$, so we may assume that for every $x \in X'$ at least two
of the sets $N_1(x), N_2(x), N_3(x)$ are non-empty. Let
$X_1$ to be the set of vertices $x \in X'$ for which $N_2(x)$ is not complete
to $N_3(x)$; for every $x \in X_1$ fix $n_2^1(x) \in N_2(x)$ and 
$n_3^1(x) \in N_3(x)$ such that $n_2^1(x)$ is non-adjacent to $n^1_3(x)$.
Define $X_2$ and $ n_1^2(x), n_3^2(x)$  for every $x \in X_2$, and
$X_3$ and $ n_1^3(x), n_2^3(x)$  for every $x \in X_3$ similarly.
Since no vertex of $X'$ is connected, it follows that
$X'=X_1 \cup X_2 \cup X_3$.

\bigskip

\noindent \textit{(1) Let $\{i,j,k\}= \{1,2,3\}$. There do not exist
$x,y \in X_i$,  $n_j \in N_j(x)$ and $n_k \in N_k(x)$ 
such that $n_j$ is non-adjacent to $n_k$, and $\{x,n_j,n_k\}$ is anticomplete
to $\{y, n_j^i(y)\}$, and $n_k^i(y)$ is anticomplete to $\{n_j,n_k\}$.}

\bigskip

\noindent Proof: Write $n_j(y)=n^i_j(y)$, and
$n_k(y)=n^i_k(y)$. By the third assumption of the theorem, there exist
$a, b \in V(G)$  such that $n_j(y)-a-b$ is a path where both $a$ and $b$ are anticomplete to $S_k$ with $|L(a)|=|L(b)|=1$. Since $x,y \in X'$, it follows that $\{a,b\}$ is anticomplete to $\{x,y\}$. 
If $x$ is adjacent to $n_k(y)$, then $n_k-x-n_k(y)-y-n_j(y)-a-b$
is a $P_7$ in $G$, a contradiction, so $x$ is non-adjacent to $n_k(y)$.
Now by the second assumption of the theorem there exists a path
$P$ with ends $a,b \in \{n_j,n_j(y),n_k,n_k(y)\}$, 
such that $\{a,b\} \neq \{n_j,n_k\}$,  $\{a,b\} \neq \{n_j(y),n_k(y)\}$,  
every interior vertex $w$  of $P$ has $|L(w)|=1$,  and 
$V(P) \setminus \{a,b\}$ is disjoint from and anticomplete to 
$\{n_j,n_j(y), n_k, n_k(y)\} \setminus \{a,b\}$.  Since $x,y \in X'$,
it follows that $V(P) \setminus \{n_j,n_j(y),n_k,n_k(y)\}$ is anticomplete
to $\{x,y\}$. But now $G[V(P) \cup \{x,y,n_j,n_j(y),n_k,n_k(y)\}]$ is
a path of length at least $7$, a contradiction. This proves \textit{(1)}.

\bigskip

Let $\{i,j,k\}= \{1,2,3\}$.  A coloring $c$ of a restriction $(G,L'',Z'')$ 
of $(G,L,Z)$ is a 
\textit{a type I coloring with respect to $i$}
if there exists $x \in X_i$, $n_j \in N_j(x)$
and $n_k \in N_k(x)$ such that  $c(n_j)=c(n_k)=i$.

\bigskip

\noindent \textit{(2) Let $(G,L'',Z'')$ be a restriction of $(G,L,Z)$.
If $(G,L'',Z'')$ admits a type I coloring with respect to $i$,
then there exists a set $\mathcal{L}_i$ of $O(|V(G)|^3)$ subpalettes of $L''$ 
such  that}

\bigskip

\noindent \textit{\textit{(a)} $|L'(v)| \leq 2$ for every $L'\in \mathcal{L}_i$ and $v \in X_i$, and }

\bigskip
  
\noindent \textit{\textit{(b)} $(G,L'',Z'')$ admits a type I coloring with respect  to
$i$ if and only if $(G,\mathcal{L}_i,Z'')$ is colorable.}

\bigskip

\noindent \textit{Moreover, $\mathcal{L}_i$ can be constructed in time $O(|V(G)|^4)$.}

\bigskip 

\noindent Proof: For every $x \in X_i$, $n_j \in N_j(x)$, $n_k \in N_k(x)$ such that
$n_j$ is non-adjacent to $n_k$,  and $c_1 \in \{j,k\}$ do the following.

Initialize the order $3$ palette $L_{x,n_j,n_k,c_1}$ of $G$:

\begin{itemize}

\item $L_{x,n_j,n_k,c_1}(x)=\{c_1\}$, 

\item $L_{x,n_j,n_k,c_1}(n_j)=L_{x,n_j,n_k,c_1}(n_k)=\{i\}$, and

\item $L_{x,n_j,n_k,c_1}(v)=L''(v)$ for all $v \in V(G) \setminus \{x\}$.

\end{itemize}

Assume that $c_1=j$; we perform a symmetric construction if $c_1=k$.
For every $y \in X_i \setminus \{x\}$ we modify $L_{x,n_j,n_k,c_1}$ as follows:
$$L_{x,n_j,n_k,c_1}(y) = \begin{cases}
 L_{x,n_j,n_k,c_1}(y)\setminus \{i\} & \text{, \hspace{2ex}if $y$ is adjacent to 
one of $n_j, n_k$, or $n_k^i(y)$ is adjacent to $x$}\\
 
 L_{x,n_j,n_k,c_1}(y)\setminus \{j\} & \text{, \hspace{2ex}if $y$ is adjacent to $x$,
or $n_k^i(y)$ is adjacent to one of $n_j, n_k$}\\

 L_{x,n_j,n_k,c_1}(v)\setminus \{k\} & \text{, \hspace{2ex}if $n_j^i(y)$ is adjacent to 
one of $n_j, n_k$}\\

\end{cases}$$

Now \textit{(1)} implies that  $|L_{x,n_j,n_k,c_1}(y)| \leq 2$ for every 
$y \in X_i$. Let  $\mathcal{L}_i$ be the set of all the $O(|V(G)|^3)$ palettes $L_{x,n_j,n_k,c_1}$ 
thus 
constructed. By construction, if $(G,\mathcal{L},Z'')$ is colorable then 
$(G,L'',Z'')$ has a type I coloring with respect to $i$. 

Now, suppose $c$ is a type I coloring of $(G,L'',Z'')$ with respect to $i$, 
and so for
some $x \in X_i$, there exist 
$n_j \in N_j(x)$ and $n_k \in N_k(x)$ with  $c(n_j)=c(n_k)=i$. 
Then $n_j$ is non-adjacent to $n_k$.
We may assume that $c(x)=j$. Then $c(x) \in L_{x,n_j,n_k,j}(x)$.
Consider a vertex $y\in X_i\setminus \{x\}$.  If 
$y$ is adjacent to  one of $n_j, n_k$, then $c(y) \neq i$.
If  $n_k^i(y)$ is adjacent to $x$, then, since $n_k^i(y) \in S_k$,
it follows that $c(n_k^i(y))=i$, and again $c(y) \neq i$.
If $y$ is adjacent to $x$, then $c(y) \neq j$. If $n_k^i(y)$ is adjacent to 
one of $n_j, n_k$, then, since $n_k^i(y) \in S_k$,
it follows that $c(n_k^i(y))=j$, and again $c(y) \neq j$.
Finally,   if $n_j^i(y)$ is adjacent to one of $n_j, n_k$,
then, since  $n_j^i(y) \in S_j$,
it follows that $c(n_j^i(y))=k$, and again $c(y) \neq k$.
Thus, in all cases, $c(y) \in L_{x,n_j,n_k,c_1}(y)$, and \textit{(2)} follows.
This proves \textit{(2)}.

\bigskip

\noindent \textit{(3) Let $(G,L'',Z'')$ be a restriction of $(G,L,Z)$.
If $(G,L'',Z'')$ does not admit a type I coloring with 
respect to either of $i,j$, then there exists a subpalette $M_{i,j}$ of $L''$
such that }

\bigskip

\noindent \textit{\textit{(a)} $|M_{i,j}(x)| \leq 2$ for every $x \in X_i \cap X_j$, and}

\bigskip
  
\noindent \textit{\textit{(b)} $(G,L'',Z'')$ is colorable if and only if $(G,M_{i,j},Z'')$ is 
colorable.}

\bigskip

\noindent \textit{Moreover, $M_{i,j}$ can be constructed in time $O(|V(G)|^2)$.}

\bigskip

\noindent Proof: For every  $x \in X_i \cap X_j$, set $M_{i,j}(x)=\{i,j\}$.
Clearly $|M_{i,j}(v)| \leq 2$ for every $x \in X_i \cap X_j$,  and 
if $(G,M_{i,j},Z'')$ is colorable, then $(G,L'',Z'')$ is colorable. Suppose that
$(G,L'',Z'')$ is colorable, and let $c$ be a coloring of $(G,L'',Z'')$. Suppose
that $c(x) \not \in M_{i,j}(x)$ for some $v \in V(G)$. Then 
$x \in X_i \cap X_j$, and $c(x)=k$.
Therefore $c(n_i^j(x))=j$ and $c(n_j^i(x))=i$. Since
$(G,L'',Z'')$ does not admit a type I coloring with respect to $i$, 
it follows that $c(n_k^i(x))=j$, but then $c$ is a type I coloring of $(G,L'',Z'')$ with respect to $j$, a contradiction. This proves \textit{(3)}.

\bigskip

\noindent \textit{(4) 
Let $(G'',L'',Z'')$ be a restriction of $(G,L,Z)$.
Suppose $(G'',L'',Z'')$ does not admit a type I coloring with
respect to $i$. Let $Y_i$ be the set of vertices $x \in X_i$ such that
$N_i(x) = \emptyset$. Let $Z_i=\bigcup_{x \in Y_i}\{N_j(x),N_k(x)\}$.
Then $(G'',L'',Z'')$ is colorable if and only if $(G''\setminus Y_i, L'', Z'' \cup Z_i)$
is colorable and a $3$-coloring of $(G''\setminus Y_i, L'', Z'' \cup Z_i)$ can be extended to a $3$-coloring of $(G'',L'',Z'')$ in time $O(|V(G)||Y_i|)$.}

\bigskip

\noindent Proof: It is enough to prove that for every coloring $c$ of 
$(G,L,Z)$ and every $x \in X_i$ such that $N_i(x)= \emptyset$, the sets 
$N_j(x)$ and
$N_k(x)$ are monochromatic with respect to $c$. Suppose not, we may assume
for some coloring $c$ there are vertices $u,v \in N_j(x)$ with $c(u)=i$
and $c(v)=k$. Since $c$ is not a type I coloring of $(G,L,Z)$,
it follows that $c(w)=j$ for every $w \in N_k(x)$. But then
$x$ has neighbors of all three colors, contrary to the fact that $c$ is a 
coloring. This proves \textit{(4)}.

\bigskip

We now construct $\mathcal{P}$ as follows. We break the construction into
four steps $\mathcal{P}_1,\mathcal{P}_2,\mathcal{P}_3$ and $\mathcal{P}_4$.

To construct $\mathcal{P}_1$,  apply \textit{(2)}
to $(G,L,Z)$ with $i=1$, to construct $\mathcal{L}_1$.
Now apply \textit{(2)}
to $(G,L',Z)$ for every $L' \in \mathcal{L}_1$ with $i=2$, to construct 
$\mathcal{L}_{12}$. Next apply \textit{(2)}
to $(G,L',Z)$ for every $L' \in \mathcal{L}_{12}$ with $i=3$,
to construct $\mathcal{L}_{123}$. Then $|\mathcal{L}_{123}|=O(|V(G)|^9)$;
by \textit{(2)} this takes time  $O(|V(G)|^{10})$.
Let $\mathcal{P}_1$ consist of all $(G,L',Z)$ with $L' \in \mathcal{L}_{123}$.

Next we construct $\mathcal{P}_2$.
Apply \textit{(4)} to $(G,L',Z)$ for every $L' \in \mathcal{L}_{12}$ with 
$i=3$; this creates a set $\mathcal{P}_2$ of $O(|V(G)|^6)$ triples 
$(G \setminus Y_3, L', Z\cup Z_3)$, and $|Z \cup Z_3|=|Z|+O(|V(G)|)$. This step can be 
performed in time $O(|(V(G)|^2)$ for every $L' \in \mathcal{L}_{12}$, and so 
takes time  $O(|(V(G)|^8)$ in total. 

Next we construct $\mathcal{P}_3$.
Apply \textit{(3)} to $(G,L',Z)$ for every $L' \in \mathcal{L}_1$ with 
$i=2$ and $j=3$; this generates a set $\mathcal{P}'_3$ of 
$O(|V(G)|^3)$ triples $(G, M',Z)$, and takes time $O(|(V(G)|^5)$.  
Now apply \textit{(4)} to every $(G,M',Z) \in \mathcal{P}'_3$ 
with $i=2$; this creates a set $\mathcal{P}_3''$ of $O(|V(G)|^3)$ triples 
$(G \setminus Y_2, M', Z\cup Z_2)$, and $|Z \cup Z_2|=|Z|+O(|V(G)|)$. This step can be 
performed in time $O(|(V(G)|^5)$. Now apply  \textit{(4)} to 
every $(G \setminus Y_2,M',Z \cup Z_2) \in \mathcal{P}_3''$ with 
$i=3$; this creates a set $\mathcal{P}_3$ of $O(|V(G)|^3)$ triples 
$(G \setminus (Y_2 \cup Y_3), M' , Z\cup Z_2 \cup Z_3)$, and 
$|Z \cup Z_2 \cup Z_3|=|Z|+O(|V(G)|)$. This step can be 
performed in time $O(|(V(G)|^5)$. 

Finally, apply \textit{(3}) to $(G,L,Z)$ with $i=1,j=2$ to obtain
$(G,M_{12},Z)$. Next apply \textit{(3}) to $(G,M_{12},Z)$ with
$i=2,j=3$ to obtain $(G,M_{12}',Z)$.  Next apply \textit{(3}) to $(G,M_{12}',Z)$ with $i=1,j=3$ to obtain $(G,M_4,Z)$. 
Now apply \textit{(4)} to with  $i=1, 2,3$ to construct 
$\mathcal{P}_4=\{(G \setminus (Y_1 \cup Y_2 \cup Y_3), M_4, Z \cup Z_1 \cup Z_2 \cup Z_3)\}$. This step takes time  $O(|(V(G)|^2)$. 

Let $\mathcal{P}'=\mathcal{P}_1 \cup \mathcal{P}_2 \cup \mathcal{P}_3 \cup \mathcal{P}_4$. Then $|\mathcal{P}'|=O(|V(G)|^9)$, and it can be constructed in time
$O(|V(G)|^{10})$. Finally, repeat the construction described above for every
permutation of the colors $\{1,2,3\}$ and let $\mathcal{P}$ be the union
of the $3!$  sets of restrictions thus generated. It is still true that
$|\mathcal{P}|=O(|V(G)|^9)$, and it can be constructed in time $O(|V(G)|^{10})$.
Moreover, by the construction process and \textit{(4)}, a $3$-coloring of a restriction in $ \mathcal{P} $ can be extended to a $3$-coloring of $G$ in time
$O(|V(G)|^2) $.

\bigskip

\noindent \textit{(5) $\mathcal{P}$ satisfies \textit{(a)}}.

\bigskip

\noindent Proof: It is enough to prove the result for $\mathcal{P}'$.
By \textit{(3)}, $|Z'|=|Z|+O(|V(G)|)$ for
every $(G',L',Z') \in \mathcal{P}$. It remains to show that
$|L'(x)| \leq 2$  for every   $(G',L',Z') \in \mathcal{P}$ and $x \in X$.

Since $X=X_1 \cup X_2 \cup X_3$, \textit{(2)} implies that
$|L'(x)| \leq 2$ for every $x \in X$ and $(G,L',Z) \in \mathcal{P}_1$.

We now check the members of $\mathcal{P}_2$. 
Also by \textit{(2)}, $|L'(x)| \leq 2$ for every $x \in X_1 \cup X_2$
and every $L' \in \mathcal{L}_{12}$. Since no vertex of $X$ is connected,
it follows that every $x \in X'$ with all three of $N_1(x),N_2(x),N_3(x)$
non-empty belongs to $X_i$ for at least two values of $i$, and so 
if $x \in X' \setminus (X_1 \cup X_2)$, then $x \in Y_3$. Since
$V(G') = V(G) \setminus Y_3$ for every 
$(G',L',Z') \in \mathcal{P}_2$, it follows that  $|L'(x)| \leq 2$ for every 
$x \in X \cap V(G')$ and $(G',L',Z') \in \mathcal{P}_2$.

Next we check the members of $\mathcal{P}_3$. By~(2),  $|L'(x)| \leq 2$ for 
every $x \in X_1$ and every $L' \in \mathcal{L}_1$. By \textit{(3)},
$|L'(x)| \leq 2$ for  every $x \in X_1 \cup (X_2 \cap X_3)$  and
every $(G,M',Z) \in \mathcal{P}_3'$.  Since no vertex of $X$ is connected,
it follows that every $x \in X'$ with all three of $N_1(x),N_2(x),N_3(x)$
non-empty belongs to $X_i$ for at least two values of $i$, and so 
if $x \in X' \setminus (X_1 \cup (X_2 \cap X_3))$, then $x \in Y_2 \cup Y_3$. 
Since $V(G') = V(G) \setminus (Y_2 \cup Y_3)$ for every 
$(G',M',Z') \in \mathcal{P}_3$, it follows that  $|L'(x)| \leq 2$ for every 
$x \in X \cap V(G')$ and $(G',M',Z') \in \mathcal{P}_3$.

Finally, we check $(G \setminus (Y_1 \cup Y_2 \cup Y_3), M_4, Z \cup Z_1 \cup Z_2 \cup Z_3)$. By \textit{(3)}, $|M_4(x)| \leq 2$ for every 
$x \in (X_1 \cap X_2) \cup (X_2 \cap X_3) \cup (X_1 \cap X_3)$.
In  particular $|M_4(x)| \leq 2$ for every $x \in X'$ with
all three of $N_1(x),N_2(x),N_3(x)$ non-empty, and so if
$x \not \in  (X_1 \cap X_2) \cup (X_2 \cap X_3) \cup (X_1 \cap X_3)$, then
$x \in Y_1 \cup Y_2 \cup Y_3$. This proves \textit{(5)}.

\bigskip

\noindent \textit{(6) $\mathcal{P}$ satisfies \textit{(b)}}.

\bigskip

\noindent Proof: Suppose first that $G$ admits a type I coloring with respect 
to each  of $1,2$ and $3$. Then by \textit{(2)},
some  $(G',L',Z') \in \mathcal{P}_1$ is colorable.

Next suppose that $G$ admits a type I coloring with respect 
to each each of $1,2$ and not with respect to $3$. 
By~\textit{(2)}, $(G,L',Z)$ is colorable for some  $L' \in \mathcal{L}_{12}$; 
now by \textit{(4)}  $(G \setminus Y_3,L',Z \cup Z_3) \in \mathcal{P}_2$ is 
colorable.

Next suppose that $G$ admits a type I coloring with respect to 1, but not with
respect to 2 or 3. By \textit{(2)}, $(G,L',Z)$ is colorable for some  
$L' \in \mathcal{L}_1$. By~\textit{(3)}, there is
$(G,M',Z) \in \mathcal{P}'_3$ that is colorable.
Now by \textit{(4)}  
$(G \setminus (Y_2 \cup Y_3) , M', Z\cup Z_3 \cup Z_3) \in \mathcal{P}_3$
is colorable.

Finally, suppose that $G$ does not admit a type I coloring with respect to any
of $1,2,3$. Now by \text{(3)} and \textit{(4)} 
$(G \setminus (Y_1 \cup Y_2 \cup Y_3), M, Z \cup Z_1 \cup Z_2 \cup Z_3) \in 
\mathcal{P}_4$ is colorable. Since we performed the same construction
for all permutation of colors $\{1,2,3\}$, this proves \textit{(6)}.

\bigskip 

Now \ref{lemma} follows from \textit{(5)} and \textit{(6)}. 
\end{proof}

\section{Coloring Expansion}

In this section, we show how to expand the set of palettes constructed in \ref{Reduce}, yielding an equivalent polynomial sized collection of sub-problems all of which can be checked by applying \ref{checkSubsets}.

\begin{theorem}\label{MainLemma}

Let $G$ be a connected $P_7$-free graph, and $A=(A_1,A_2,A_3)$ be a 
normal tripod in $G$, and assume that $G$ is $(A_1,A_2,A_3)$-clean. Partition 
$V(G)=A\cup X\cup Y\cup Z$ as in \ref{Reduce}. 
Let $\mathcal{L}$ be the set of palettes generated by \ref{Reduce} and consider a fixed palette $L\in \mathcal{L}$. Then there exists a set  $\mathcal{P}_L$ of $O(|V(G)|^{9})$ restrictions  of $(G,L,\emptyset)$ such that the following hold:

\bigskip

\noindent \textit{(a) For every $(G',L',S)\in \mathcal{P}_L$, $|L'(v)| \leq 2$ for every $v \in V(G')$  and $ |S|= O(|V(G)|) $ , and}

\bigskip
  
\noindent (b) $(G,L)$ is colorable if and only if $ \mathcal{P}_L$ is colorable.

\bigskip

\noindent \textit{Moreover, $\mathcal{P}_L$ can be constructed in time $O(|V(G)|^{10})$, and a $3$-coloring of a restriction in $ \mathcal{P}_L $ can be extended to a $3$-coloring of $G$ in $ O(|V(G)|^2) $.}

\end{theorem}

\begin{proof}
We use the notation of~\ref{Reduce}.
By~\ref{Reduce}, for every $x \in P_L$, 
$N(x) \subseteq (X\setminus X') \cup Y'$, and $|L(v)| \leq 2$ for every 
$v \in N(x)$. 

Let $\{i,j,k\}=\{1,2,3\}$. We remind the reader that by~\ref{Reduce} 

\begin{itemize}
\item \textit{(a)} If $v\in Y'$ with $L(v)=\{i,j\}$, then there exists $u\in N(v)\cap (X'\cup Y_0^s)$ with $L(u)=\{k\}$
\item  \textit {(b)} If $v\in X'\cap X_j$ with $L(v)=\{i\}$, then either there exists $u\in N(v)$ such that $L(u)=\{k\}$, or $L(y)=\{j\}$ for every
$y \in Y$ with a neighbor in $X_j$, and
\item  \textit{(c)} If $v\in Y^s_0$ with $L(v)=\{i\}$, then there exists $u,w\in N(v) \cap \{s_1,s_2,s_3\}$ such that $L(u)=\{k\}$ and $L(w)=\{j\}$. 
\end{itemize}

Next we repeatedly update $L$ until we perform a round of updating in which
no list is changed. This requires at most $|V(G)|$ rounds of updating, and so
takes time $O(|V(G)|^3)$. Now let $P$ be the set of vertices $v \in P_L$ 
with  $|L(v)|=3$. By updating, we may assume that for every $v \in P$
and for every neighbor $y$ of $v$, we have $|L(y)|=2$.
For $1 \leq i <j \leq 3$
and $k \in \{1,2,3\} \setminus \{i,j\}$,
let $S_k$ be the set of vertices 
$v \in  (X\setminus X') \cup Y'$ such that $v$ has a neighbor in 
$P$, and $L(v)=\{i,j\}$. Since we have updated, it follows that
every vertex $w$ with $L(w) \in \{\{i\},\{j\}\}$ is anticomplete to $S_k$.

It is now enough to check that $S_1,S_2,S_3,P$ satisfy the assumptions
of~\ref{lemma} (where $P$ plays the roles of $X$ from \ref{lemma}). Since 
every vertex of $P$ is anticomplete to 
$A_1 \cup A_2 \cup A_3$ it follows that no vertex of $P$ is connected.
By definition, the lists of $S_1,S_2,S_3$ satisfy the first 
condition.

Now we check the second condition.  Let $1 \leq i <j \leq 3$ and let 
$u_i,v_i \in S_i$ and $u_j,v_j \in S_j$ such that $\{u_i,v_i,u_j,v_j\}$
is a stable set. We may assume $i=1$ and $j=2$. Then
$u_1,v_1 \in X_1 \cup Y'$ and $u_2,v_2 \in X_2 \cup Y'$.
Suppose first that both $u_1,v_1 \in X_1$. By~\ref{tricomps}, there is a path 
$P$ from $u$ to  $v$ with interior in $A_1 \cup A_3$. Since 
$u_2,v_2 \in S_2$, it follows that the interior of $P$ is 
anticomplete to and disjoint from $\{u_2,v_2\}$, as required.

Next suppose that $u_1 \in X_1$.  Then $v_1 \in Y'$, and therefore 
$v_1$ is anticomplete to $A_1 \cup A_2 \cup A_3$.
Assume first that $v_2 \in X_2$. Then
$u_2 \in Y'$, and in particular, $u_2$ is anticomplete
to $A_1 \cup A_2 \cup A_3$. Let $P$ be a path from $u_1$ to $v_2$
with interior in $A_1 \cup A_2$ (which exists by~\ref{tricomps});
then $P$ has the required properties. Thus we may assume that 
$v_2\in Y'$. 
By \textit{(a)}, there exists $w \in X \cup Y_0$ such that $v_2$
is adjacent to $w$, and $L(w)=\{2\}$. Then $w$ is anticomplete to 
$\{u_1,v_1\}$.  We may also assume $w$ is  anticomplete to $\{u_2\} $ since other wise $u_2-w-v_2$ is the desired path. If $w \in X_1\cup X_3$, 
then by~\ref{tricomps} there is a path $P$ from $u_1$ to $w$ with interior in
$A_1 \cup A_3$, and $u_1-P-w-v_2$ is the desired path. 
So we may assume that $w \in Y_0$. Then $L(s_1) = \{3\}$, since $s_3$ is adjacent to $w$, $L(w)=\{2\}$ and $s_1\in X_1$. Hence 
 $s_1$ is anticomplete to $\{u_1,u_2,v_1,v_2\}$. By~\ref{tricomps},
there is a path $P$ from $s_3$ to $u_1$ with interior in $A_1 \cup A_3$.
But now $v_2-w-s_1-P-u_1$ is the required path.

Thus we may assume that $u_1,u_2,v_1,v_2 \in Y$.
Let $a\in N(u_1)\cap(X'\cup Y_0^s) $ and $b\in N(v_1)\cap (X'\cup Y_0^s)$ 
with $L(a)=L(b)=\{1\}$.  Such $a,b$ exist by \textit{(a)}. 
Then $\{a,b\}$ is anticomplete to 
$\{u_2,v_2\}$. If there is a path $P$ from 
$a$ to $b$ with (possibly empty) interior in $A_1 \cup A_2 \cup A_3$, then
$u_1-a-P-b-v_1$ is the desired path, so we may assume no such path $P$ exists.
It follows that $a \neq b$, $a$ is non-adjacent to $b$, and at least one of 
$a,b$ belongs to $Y_0$. We may assume that $a \in Y_0$.  Therefore
$L(s_2)=\{3\}$, and so $s_2$ is anticomplete to $\{u_2,v_2\}$.
If $b$ is adjacent to some $s_2$, then $u_1-a-s_2-b-v_1$
is the desired path, so we may assume not. It follows that $b \in X$.
By \ref{tricomps} there is a path from $s_2$ to $b$ with interior
in $A_1 \cup A_2 \cup A_3$, and now  $u_1-a-s_2-P-b_1$ is the desired path.
Thus the second condition holds.

Lastly, we verify that the third condition holds. Let $i,j \in \{1,2,3\}$
and let $k \in \{1,2,3\} \setminus \{i,j\}$.  Consider $u\in S_i$. 

We claim that $u$ has a neighbor $a$ with $L(a)=\{i\}$,
and $a$ has a neighbor $b$ with $L(b)=\{k\}$, and $u-a-b$ is a path.
Suppose first that $u \in X_i$. Then $u$
has a neighbor $a \in A_i$, and $a$ has a neighbor $b \in A_k$, as required.
Thus we may assume that $u \in Y'$. 
Since $L(u)=\{j,k\}$, by \textit{(a)}, there exists    
$a\in N(u)\cap (X'\cup Y_0^s)$ with list $\{i\}$.  
Since $a$ has list $\{i\}$, it follows that $a \in X_j \cup X_k \cup Y_0$.
By \textit{(b)}  and \textit{(c)}, and since every vertex of $X_k$ has a 
neighbor in $A_k$, it follows that $a$ has a neighbor $b$ with 
$L(b)=\{k\}$.  Since $L(u)=\{j,k\}$ and we have updated, it follows
that $b$ is non-adjacent to $u$, and $u-a-b$ is a path.
This proves the claim.

Since $L(v)=\{i,k\}$ for every $v \in S_j$, and since we have update,
it follows that $\{a,b\}$ is anticomplete to $S_j$ as required.
Thus the third condition holds. This proves~\ref{MainLemma}. 
\end{proof}

\section{Main Result}

In this section we prove the main result of this paper \ref{half2}, which we restate:

\begin{lemma}\label{half2'}There is an algorithm with the following specifications:

\bigskip

{\bf Input:} A $P_7$-free graph $G$ which contains a triangle.

\bigskip

{\bf Output:} A $3$-coloring of $G$, or a determination that none exists.

\bigskip

{\bf Running time:} $O(|V(G)|^{24})$.

\end{lemma}

\begin{proof}
We may also assume that $G$ is connected (otherwise we run
the following procedure for each connected component of  $G$). 
 By \ref{tripod}, at the expense of carrying out a time $O(|V(G)|^3)$ procedure we can determine that no $3$-coloring of $G$ exists (then we can stop), or obtain a connected graph $G'$ satisfies the following:
\begin{itemize}
\item $|V(G')| \leq |V(G)|$,
\item $G'$ is connected,
\item $G'$ is $3$-colorable if and only if $G$ is $3$-colorable,
\item Any 3-coloring of $G'$ can be extended to a 3-coloring of $G$ in time $O(|V(G)|^2)$, and
\item $G'$ is either triangle-free or contains a normal tripod $(A_1,A_2,A_3)$.
\end{itemize}
In the case that $G'$ is triangle-free we can use the algorithm in \cite{Tfree} to either determine that no $ 3 $-coloring of $G'$ exists or find a $3$-coloring of $G'$ in $ O(|V(G)|^7)$. Thus we can either determine that no $3$-coloring of $G$ exists or use the $3$-coloring of $G'$ to find a $3$-coloring of $G$ in time $O(|V(G)|^2)$. 

Thus we may assume that $G'$ contains a normal tripod $(A_1,A_2,A_3)$. 
By \ref{cleaning}, at the expense of carrying out a time $O(|V(G)|^4)$ procedure, we can either determine that $G$ is not $3$-colorable (and stop), or 
may assume that $G'$ is $(A_1,A_2,A_3)$-clean. 
By \ref{Reduce}, in time $ O(|V(G)|^{15}) $ we can produce a set $\mathcal{L}$ of $O(|V(G)|^{12})$  order $3$ palettes  of $G'$ such that $G'$ has a $3$-coloring if and only if $(G',\mathcal{L})$ is colorable. By \ref{MainLemma} for a fixed $L\in \mathcal{L}$, in time $ O(|V(G)|^{10}) $ we can construct a set of $O(|V(G)|^{9})$ restrictions $\mathcal{P}_L$ such that  

\begin{itemize}
\item For every $(G'',L',X)\in \mathcal{P}_L$, $|L'(v)| \leq 2$ for every $v \in V(G'')$  and $ |X|= O(|V(G)|) $,
\item $(G',L)$ is colorable if and only if $ \mathcal{P}_L$ is colorable , and
\item a $3$-coloring of a restriction in $ \mathcal{P}_L $ can be extended to a $3$-coloring of $G'$ in time $ O(|V(G)|^2) $.
\end{itemize}
For every restriction in $\mathcal{P}_L$, by \ref{checkSubsets}, in time $O(|V(G)^3|)$ we can either determine that it is not colorable, or find a coloring of it. Since $|\mathcal{L}|= O(|V(G)|^{12})$ and $|\mathcal{P}_L|=O(|V(G)|^{9})$, we need to run \ref{checkSubsets} $O(|V(G)|^{21})$ times. Hence in time $O(|V(G)|^{24})$, we can either determine that no $ 3 $-coloring of $G'$ exists, which means that no $ 3 $-coloring of $G$ exists, or find a $3$-coloring of $G'$, which can be extended to a $3$-coloring of $G$ in time $O(|V(G)|^{2})$.  This proves \ref{half2'}. 
\end{proof}

\section{Acknowledgment}
We are grateful to Alex Scott for telling us about this problem, and to Juraj 
Stacho for sharing his knowledge of the area with us.

\end{document}